\documentclass[10pt,journal,twocolumn]{IEEEtran}
\usepackage{amssymb,amsmath}
\newtheorem{Theorem}{Theorem}

\newtheorem{Lemma}{Lemma}

\usepackage{cite}
\usepackage{graphics}
\usepackage{graphicx}

\begin{document}
\title{Distributed Linear Convolutional Space-Time Coding for Two-Relay Full-Duplex Asynchronous Cooperative Networks\thanks{Y. Liu and H. Zhang are with the State Key Laboratory of Integrated Service Network, Xidian University, Xi'an, 710071, China (e-mail: \{yliu, hlzhang\}@xidian.edu.cn). Their research was supported in part by the National Natural Science Foundation of China (61301169), the Fundamental Research Funds for the Central Universities (72125377), the Important National Science \& Technology Specific Projects (2012ZX03003012-003), the 111 Project (B08038), China.
X.-G. Xia is with the Department of Electrical and Computer Engineering, University of Delaware, Newark, DE 19716, USA (e-mail: xxia@ee.udel.edu). His research was supported in part by the National Science Foundation (NSF) under Grant CCF-0964500.}
}

\author{Yi Liu, Xiang-Gen Xia, and Hailin Zhang}

\maketitle

\begin{abstract}
In this paper, a two-relay full-duplex asynchronous cooperative network with the amplify-and-forward (AF) protocol is considered. We propose two distributed space-time coding schemes for the cases with and without cross-talks, respectively. In the first case, each relay can receive the signal sent by the other through the cross-talk link. We first study the feasibility of cross-talk cancellation in this network and show that the cross-talk interference cannot be removed well. For this reason, we design space-time codes by utilizing the cross-talk signals instead of removing them. In the other case, the self-coding is realized individually through the loop channel at each relay node and the signals from the two relay nodes form a space-time code. The achievable cooperative diversity of both
cases is investigated and the conditions to achieve full cooperative diversity are presented. Simulation results verify the theoretical analysis.

\end{abstract}


\IEEEpeerreviewmaketitle

\section{Introduction}

In cooperative wireless communication networks,
multiple nodes work together to form a virtual
multi-input and multi-output (MIMO) system.
Using cooperation, it is possible to exploit the
spatial diversity similar to a MIMO system,
 see for example, \cite{Politis}--\cite{1321222}.
According to the time slots of receiving and transmitting, the working modes at relay nodes can be categorized into full-duplex (FD) and half-duplex (HD) modes.
With  the HD mode a relay receives and transmits signals on orthogonal
(in time or frequency) channels, while
with the FD mode it uses only one channel
\cite{5161790,5089955,5470111,5961159,4917875,6069620}.
Thus, an FD cooperative protocol may achieve a higher bandwidth efficiency than
an HD cooperative protocol \cite{1499041}. However, the FD mode introduces loop (self)
 interference due to the signal leakage between the same relay's output and input and sometimes cross-talk interference among different relays' output and input.
%
To deal with this problem and analyze the feasibility of the capacity
gain with the FD mode, recent efforts have been made
in \cite{5161790, 5089955, 5470111, 5961159, 4917875,6069620}
where various loop interference cancellation schemes have been proposed
for networks with one relay, and in \cite{6151271} with two or more relays where
in addition to self-loop interference, cross-talk interference between the relays
may occur. In terms of cross-talk interference, the case with multiple single-input single-output (SISO) relays is similar to the case with one multiple-input multiple-output relay. However, for the case with one MIMO relay, the relay can know the cross-talk exactly since it is sent by itself and also the synchronization is not a problem. On the other hand, for the case with multiple SISO relays, one relay does not know the cross-talk interference exactly since it is sent by other relays and furthermore, the signals may not be synchronized among all the relays.

In our previous work \cite{Yiliu}, we proposed a different way to deal with the
self-loop interference for a cooperative network with one relay node, where
not all the loop interference is cancelled but instead some of them are utilized as the coding
(space-time coding) to achieve the spatial diversity. Since there is only one relay used,
the cross-talk interference is not an issue. In this paper, we consider a cooperative
network with two amplify-and-forward (AF) relays where cross-talk may occur. To deal with both loop interference
and cross-talk interference, we propose a partial
 distributed linear convolutional space-time coding (partial DLC-STC) scheme
where the cross-talk interference is utilized as a part of the
partial DLC-STC as the self-coding. Note that here we adopt the time domain approach but not the frequency domain, i.e., orthogonal frequency division multiplexing (OFDM), approach. This is because the signal model in this case  may induce infinite length (or the same as any block length as we shall see later) impulse responses in the equivalent channel and a long cyclic prefix would be needed for the OFDM approach. When there is no cross-talk interference, we also propose
a DLC-STC scheme where some of the loop interference
is used for the self-coding that is similar to but more general than Scheme two
proposed in \cite{Yiliu}.
In both cases, we illustrate that the proposed DLC-STC schemes can achieve full asynchronous cooperative
 diversity of two.

This paper is organized as follows. In Section \ref{sec2},
we formulate the system and signal models for two-relay two-hop FD cooperative networks, where we show that the cross-talk interference at relays cannot be removed well even when the relays know all the accurate channel state information.
In Section \ref{sec3}, we present the construction method of DLC-STC for FD
 asynchronous cooperative communications with cross-talks. In Section \ref{sec5-1}, we show the construction method when there is no cross-talk link between the relays. In Section \ref{sec6}, we present some simulation results to evaluate
the performances of the proposed schemes. Finally,
in Section \ref{sec7}, we conclude this paper.

\section{System model and motivation}\label{sec2}
\begin{figure}[htbp]
\centering
\includegraphics[scale=0.8]{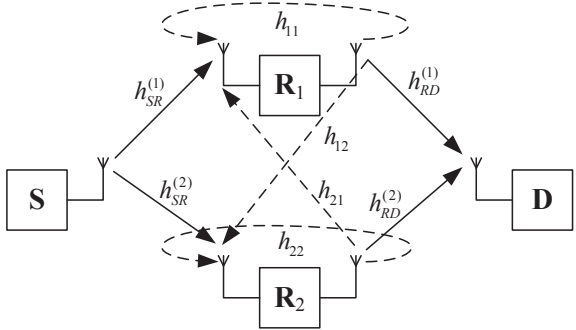}
\caption{\label{fig1}Two-relay two-hop cooperative network with potential loop and cross-talk interference.}
\end{figure}

Consider a cooperative network shown in Fig.~\ref{fig1}, where there are two single-input single-output (SISO) relay nodes between the source node with one transmit antenna and the destination node with one receive antenna.
The source node communicates with the destination node via the relay nodes and the direct link is assumed too weak to be considered.
Each relay node receives and sends signals with the same frequency band at the same time. So there is a loop link at each relay node as well as cross-talk links between the two relays. The channel from the source to the $k$th relay is $h_{SR}^{(k)}$, from the $k$th relay to
the destination is $h^{(k)}_{RD}$, and the loop channel of the $k$th relay is $h_{kk}$. Since all the relays use the same frequency band, there is a cross-talk
link $h_{jk}$ between the $j$th relay transmitter and the $k$th relay receiver. All of the channels are assumed to be
quasi-static and follow the distribution of $\mathcal{CN}(0,1)$, that is, the channels keep constant during each frame and change between frames. The delay from the source node to Relay $k$ is $\varphi_k$. The two relays are assumed not far away from each other, for which the transmission delay of the cross-talk links is too small to consider as well as the loop links. Thus, the received signal at Relay $k$ at time slot $i$ can be written as
\begin{equation}
r^{(k)}(i)=h_{SR}^{(k)}x(i-\varphi_k)+ h_{kk}t^{(k)}(i)+h_{j k}t^{(j)}(i)+w^{(k)}_R(i), j\neq k \label{equ1-1}
\end{equation}
where $t^{(k)}(i)$ is the signal sent by Relay $k$ and
 $w^{(k)}_R(i)$ is the additive noise with the distribution of $\mathcal{CN}(0,\sigma_R^2)$ at Relay $k$, and the delay $\varphi_k$  is normalized as an integer by the information symbol period $T_s$ since the fractional delay can be absorbed in the channels.
The channel state information $h_{SR}^{(k)}$, $\varphi_k$, and $h_{jk}$ are assumed to be known
at all the relays and the destination.

The second and third terms in the right hand side of equation (\ref{equ1-1}) are signals from the loop channel and the cross-talk channel, respectively, which are regarded as interference in general.
One obvious idea would be to cancel the loop interference corresponding to the second term and the cross-talk interference corresponding to the third term.
Then, after these interference terms are cancelled, the signal model
would become the same as the existing
HD model and thus
the existing cooperative relay schemes for HD model could be used.
Let us see whether this idea works.
The second term is easy to be removed since each relay knows what is sent by itself.
However, the third term is  from the other relay. To remove this, Relay $k$ can only use the signal
 estimated by itself to reproduce this term (even though all the channel state
information is known at both relays), during which the noise would be
unfortunately propagated and accumulated as we shall see in more details below.

In this paper, the amplify-and-forward (AF) protocol \cite{5089955,5961159,4917875} is considered. Accordingly, the transmission signal at Relay $k$ is
\begin{equation}\label{equ2-10}
t^{(k)}(i)=\beta_{k}\mathcal{L}(r^{(k)}(i-\phi))
\end{equation}
where $\beta_{k}$ is the amplifying factor to control the transmission power to be 1, $\phi$ is the common delay at the relays,
and $\mathcal{L}(\cdot)$ denotes  some simple linear operations, such as
interference subtractions as we shall see later
and zero-forcing (ZF) or minimum mean square error (MMSE) estimator
that will be described below.
From (\ref{equ1-1}), it is clear
 that the condition that the cross-talk interference $t^{(j)}(i)$ can be cancelled is that Relay $k$ can somehow reproduce $t^{(j)}(i)$. This
implies that  the time index $i-\phi \leq i-\varphi_k-1$, for $k=1,2$.
Thus,  the common delay at the transmissions of two relays
 should be controlled such that $\phi\geq \max\{\varphi_1,\varphi_2\}+1$.
Suppose for $i\leq i_0$, $i_0=\phi$, the loop interference and cross-talk interference are removed perfectly in (\ref{equ1-1}). For simplicity, we assume
that the ZF estimation is used and then the estimated signal can be ideally
written as
\begin{equation}\label{equ2-11}
\begin{array}{rcl}
\mathcal{L}(r^{(k)}(i))&=&\frac{1}{h_{SR}^{(k)}}[r^{(k)}(i)-h_{kk}t^{(k)}(i)-h_{j k}t^{(j)}(i)]\\
&=&\hat{x}_{k}(i-\varphi_k)=x(i-\varphi_k)+\frac{w^{(k)}_R(i)}{h_{SR}^{(k)}},i\leq i_0.
\end{array}
\end{equation}
Substituting (\ref{equ2-11}) into (\ref{equ2-10}), we obtain the transmission signal at time slot $i_0+1$ as follows,
\begin{equation}\label{equ2-110}
\begin{array}{rcl}
t^{(k)}(i_0+1)&=&\beta_{k}\hat{x}_k(i_0+1-\phi)\\
&=&\beta_{k}\left[x(i_0+1-\phi)+\frac{w^{(k)}_R(i_0+1-\phi+\varphi_k)}{h_{SR}^{(k)}}\right].
\end{array}
\end{equation}
The received signal at time slot $i_0+1$ is
\begin{equation}\label{equ2-12}
\begin{array}{rcl}
r^{(k)}(i_0+1)=h_{SR}^{(k)}x(i_0+1-\varphi_k)+h_{kk}t^{(k)}(i_0+1)\\+h_{j k}t^{(j)}(i_0+1)+w^{(k)}_R(i_0+1), j\neq k.
\end{array}
\end{equation}
The loop interference can be removed perfectly since it is sent by Relay $k$ itself, but for the cross-talk $t^{(j)}(i_0+1)$, it can only be reproduced and cancelled by using the estimated symbols at Relay $k$, that is, $\hat{x}_{k}(i)$. Thus, we can obtain (\ref{equ2-13}) at the top of the next page.
\newcounter{MYtempeqncnt}
\begin{figure*}[!t]
\normalsize
\setcounter{MYtempeqncnt}{\value{equation}}
\begin{equation}\label{equ2-13}
\begin{array}{rcl}
\hat{x}_{k}(i_0+1-\varphi_k)&=&\frac{1}{h_{SR}}[r^{(k)}(i_0+1)-h_{kk}t^{(k)}(i_0+1)-h_{j k}\beta_{j}\hat{x}_k(i_0-\phi)]\\
&=&x(i_0+1-\varphi_k)+\frac{w^{(k)}_R(i_0+1)}{h_{SR}^{(k)}}+\frac{h_{jk}\beta_{j}[\hat{x}_j(i_0-\phi)-\hat{x}_k(i_0-\phi)]}{h_{SR}^{(k)}}\\
&=&x(i_0+1-\varphi_k)+\frac{w^{(k)}_R(i_0+1)}{h_{SR}^{(k)}}+\frac{h_{j k}\beta_{j}[w^{j}_R(i_0+\varphi_k-\phi)-w^{k}_R(i_0+\varphi_k-\phi)]}{h_{SR}^{(k)}}
\end{array}
\end{equation}
\setcounter{equation}{10}
\begin{equation}\label{equ11-cstc01}
\begin{array}{lll}
r^{(k)}(i)\!\!\!\!&=\sum_{n=0}^{\infty}\eta^n\left[ h_{SR}^{(k)}{x}(i-2n\phi-\varphi_k)+\!\! h_{jk}\beta_jh_{SR}^{(j)}{x}(i-(2n+1)\phi-\varphi_j)\right] \\
&+\sum_{n=0}^{\infty}\eta^n\left[{w}^{(k)}(i-2n\phi)+\!\! h_{jk}\beta_j{w}^{(j)}(i-(2n+1)\phi)\right]
\end{array}
\end{equation}
\begin{equation}\label{equ11-cstc02}
\begin{array}{rcl}
r^{(k)}(i)&=&\sum_{u=0}^{u_1-1}\eta^u h_{SR}^{(k)}{x}(i-2u\phi-\varphi_k)+\sum_{v=0}^{v_1-1} \eta^v h_{jk}\beta_jh_{SR}^{(j)}{x}(i-(2v+1)\phi-\varphi_j) \\
&&+\sum_{u=u_1}^{u_2}\eta^u h_{SR}^{(k)}{x}(i-2u\phi-\varphi_k)+ \sum_{v=v_1}^{v_2} \eta^v h_{jk}\beta_jh_{SR}^{(j)}{x}(i-(2v+1)\phi-\varphi_j) \\
&&+\sum_{u=u_2+1}^{\infty}\eta^u h_{SR}^{(k)}{x}(i-2u\phi-\varphi_k)+ \sum_{v=v_2+1}^{\infty} \eta^v h_{jk}\beta_jh_{SR}^{(j)}{x}(i-(2v+1)\phi-\varphi_j)\\
&&+\sum_{n=0}^{\infty}\eta^n\left[{w}^{(k)}(i-2n\phi)+\!\! h_{jk}\beta_j{w}^{(j)}(i-(2n+1)\phi)\right]
\end{array}
\end{equation}
\setcounter{equation}{\value{MYtempeqncnt}}
\hrulefill
\end{figure*}

The third term in the right hand side of (\ref{equ2-13}) is the noise
propagated from the estimated symbols and will be  also propagated
to  the following transmission signal defined by (\ref{equ2-10}):
\setcounter{equation}{6}
\begin{equation}\label{equ2-14}
\begin{array}{lll}
t^{(k)}(i_0\!+\!1\!+\!\phi\!-\!\varphi_k)\!\!\!\!&=\beta_{k}\hat{x}_k(i_0+1-\varphi_k)\\
&=\beta_{k}\left[x(i_0+1-\varphi_k)+\frac{w^{(k)}_R(i_0+1)}{h_{SR}^{(k)}}\right]\\
&+\frac{\beta_{k}h_{j k}\beta_{j}[w^{j}_R(i_0+\varphi_k-\phi)-w^{k}_R(i_0+\varphi_k-\phi)]}{h_{SR}^{(k)}}.
\end{array}
\end{equation}
Comparing (\ref{equ2-110}) and (\ref{equ2-14}), we notice that the noise is
accumulated in the transmission signal
and  increases along with the time index.

Although  it is hard to achieve a closed form of the accumulated noise,
 to have a qualitative observation, we calculate the average SNR of the transmission signals at the two relays by simulations and the result is shown in Fig.~\ref{fig00}.
\begin{figure}[htbp]
\centering
\includegraphics[scale=0.55]{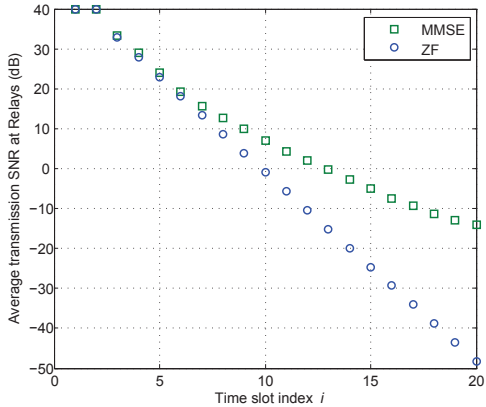}
\caption{\label{fig00}Average transmission SNR versus the time index $i$ of full-duplex two relays network.}
\end{figure}
The parameters used in Fig.~\ref{fig00} are $\varphi_1=\varphi_2=0$ and $\phi=1$.
The signal power is normalized to be 1 and the noise power $\sigma_R^2=-40$dB.
 The length of each block is $20$.
The channel coefficients remain unchanged during one block and
they  change randomly  between different trials.
We notice that with the increase in the time  index,
the transmission SNR decreases from $40$dB to $-49$dB and $-14$dB with the ZF and MMSE estimators, respectively, which means
that the desired signal is completely  buried in
the noise at the end of the block.

From the analysis above, we can see that it is impossible
for the relays to accurately estimate the symbols sent by the source node
when the AF protocol is used so that
the cross-talk interference can be removed. In other words, the cross-talk
interference cannot be removed well.
It is not hard to see that this conclusion also holds when a distributed
space-time coding is applied at the relays as we will do later.
This motivates the study for this paper, which is that since the cross-talks
cannot be removed well, why do not we use them as the coding at the relays?
In the following, we propose a space-time coding scheme for this FD mode cooperative communication network by adopting the cross-talk signals.

\section{Coding for the case with cross-talk}\label{sec3}
\subsection{Construction of partial DLC-STC}\label{sec3a}
\setcounter{equation}{17}
\begin{figure*}[!t]
\normalsize
\setcounter{MYtempeqncnt}{\value{equation}}
\begin{equation}\label{equ11-a02}
\tilde{t}^{(1)}(i)\!=\left\{\!\!\!\!\!\!\!\!\begin{array}{lr}\begin{array}{rcl}&&\!\!\!\!\sum\limits_{n=0}^{\frac{L-3}{2}}\eta^n\left[\beta_1 h_{SR}^{(1)}{x}(i-(2n+1)\phi-\varphi_1)+\!\!\beta_1 h_{21}\beta_2h_{SR}^{(2)}{x}(i-(2n+2)\phi-\varphi_2)\right] \\
&&+\sum\limits_{n=0}^{\frac{L-3}{2}}\eta^n\left[\beta_1{w}^{(1)}(i-(2n+1)\phi)+\!\!\beta_1 h_{21}\beta_2{w}^{(2)}(i-(2n+2)\phi)\right]\\
&&+\eta^{\frac{L-1}{2}}\beta_1h_{SR}^{(1)}{x}(i-L\phi-\varphi_1)+\eta^{\frac{L-1}{2}}\beta_1h_{SR}^{(1)}{w}^{(1)}(i-L\phi)\\
&&+\,\,\eta^{\frac{L-1}{2}}\beta_1h_{21}{\tilde{t}}^{(2)}(i-L\phi),
\end{array}, \,\,\text{$L$ is odd}\\
\begin{array}{rcl}&&\!\!\!\!\sum\limits_{n=0}^{\frac{L}{2}-1}\eta^n\left[\beta_1 h_{SR}^{(1)}{x}(i-(2n+1)\phi-\varphi_1)+\!\!\beta_1 h_{21}\beta_2h_{SR}^{(2)}{x}(i-(2n+2)\phi-\varphi_2)\right] \\
&&+\sum\limits_{n=0}^{\frac{L}{2}-1}\eta^n\left[\beta_1{w}^{(1)}(i-(2n+1)\phi)+\!\!\beta_1 h_{21}\beta_2{w}^{(2)}(i-(2n+2)\phi)\right]\\
&&+\,\,\eta^{\frac{L}{2}}{\tilde{t}}^{(1)}(i-L\phi),
\end{array}, \,\,\text{$L$ is even}
\end{array}\right.
\end{equation}
\begin{equation}\label{equ11-a02-2}
\tilde{t}^{(1)}(i)\!=\left\{\!\!\!\!\!\!\!\!\begin{array}{lr}\begin{array}{rcl}&&\!\!\!\!\sum\limits_{n=0}^{\frac{L-3}{2}}\eta^n\left[\beta_1 h_{SR}^{(1)}{x}(i-(2n+1)\phi-\varphi_1)+\!\!\beta_1 h_{21}\beta_2h_{SR}^{(2)}{x}(i-(2n+2)\phi-\varphi_2)\right] \\
&&+\sum\limits_{n=0}^{\frac{L-3}{2}}\eta^n\left[\beta_1{w}^{(1)}(i-(2n+1)\phi)+\!\!\beta_1 h_{21}\beta_2{w}^{(2)}(i-(2n+2)\phi)\right]\\
&&+\eta^{\frac{L-1}{2}}\beta_1h_{SR}^{(1)}{x}(i-L\phi-\varphi_1)+\eta^{\frac{L-1}{2}}\beta_1h_{SR}^{(1)}{w}^{(1)}(i-L\phi),
\end{array},\,\,\text{$L$ is odd}\\
\begin{array}{rcl}&&\!\!\!\!\sum\limits_{n=0}^{\frac{L}{2}-1}\eta^n\left[\beta_1 h_{SR}^{(1)}{x}(i-(2n+1)\phi-\varphi_1)+\!\!\beta_1 h_{21}\beta_2h_{SR}^{(2)}{x}(i-(2n+2)\phi-\varphi_2)\right] \\
&&+\sum\limits_{n=0}^{\frac{L}{2}-1}\eta^n\left[\beta_1{w}^{(1)}(i-(2n+1)\phi)+\!\!\beta_1 h_{21}\beta_2{w}^{(2)}(i-(2n+2)\phi)\right],
\end{array},\,\,\text{$L$ is even}
\end{array}\right.
\end{equation}

\setcounter{equation}{\value{MYtempeqncnt}}
\hrulefill
\end{figure*}

From (\ref{equ1-1}), we can see that the second term at the right
hand side is the signal sent by Relay $k$ itself. Since the loop channel $h_{kk}$ is  known by Relay $k$, this term can be completely cancelled from the received signal (note that in \cite{Yiliu} the self-loop interference is intentionally not removed
completely but instead part of it is maintained as the self-coding). After
the complete cancellation, the signal model can be written as follows.

\setcounter{equation}{7}
\begin{eqnarray}
r^{(k)}(i)&=&h_{SR}^{(k)}x(i-\varphi_k)+h_{jk}t^{(j)}(i)+w^{(k)}_R(i)\label{equ2-0} \\
t^{(k)}(i)&=&\beta_k r^{(k)}(i-\phi) \label{equ2},
\end{eqnarray}
where $x(i)$ is the transmitted signal by the source node with normalized power $E_s=E[|x(i)|^2]=1$ and $w^{(k)}_R(i)$ is the additive $\mathcal{CN}(0,\sigma_R^2)$ noise at the receiver of Relay $k$. During the following analysis, we assume that the parameters $\beta_1, \beta_2, \varphi_1, \varphi_2, h_{21}, h_{12}$ and $\phi$ are known at both relays.

To avoid overlap between the neighboring coded frames due to the transmission delays, a simple method is to protect the data sequence with zero guard intervals \cite{4155128,Guo}. Suppose the zero padding length is $p$ and the data sequence to be sent by the source node is $s(i)$, then the zero padded signal sent by the source node is
\begin{equation}\label{equ11-s0}
x(i)\!=\!\left\{
\begin{array}{cc}
\!\!\!s(i-mp),&\!\!\!\!\!\!\!\!m(N+p)\leq i\leq m(N+p)+N\!-\!1\\
\!\!\!0,&\!\!\!\!\!\!\!\!\!\!\!\!m(N+p)+N\leq i\leq (m+1)(N+p)\!-\!1
\end{array}\right.\!\!\!\!,
\end{equation}
where $N$ is the data frame length, $N+p$ is the frame length, and $m$ is the frame index. The zero padding length $p$ will be specialized later.

Before we describe the partial DLC-STC, let us first see what relays transmit and receive, when the source sends the above framed and zero-padded signal $x(i)$.
Without loss of generality, let us only consider the 0th frame, i.e., $m=0$ in (\ref{equ11-s0}). In this case, substituting (\ref{equ2}) into (\ref{equ2-0}) recursively,
we obtain the signal received at Relay $k$ as (\ref{equ11-cstc01}), where $\eta=\beta_1\beta_2h_{12}h_{21}$.

Next we will show that there are always non-zero symbols in the $-1$th frame involved in $r^{(k)}(i)$ no matter how many zeroes are padded in $x(l)$ in (\ref{equ11-s0}).

To do so, let us rewrite (\ref{equ11-cstc01}) as (\ref{equ11-cstc02}).
For any $i\geq0$, no matter how $p$ is,  the condition that the index of $x(l)$ in (\ref{equ11-cstc02}) is within non-zero data symbols' indices of the
$-1$th frame is
\setcounter{equation}{12}
\begin{equation}\label{equ11-cstc03}
\left\{
\begin{array}{l}
-N-p\leq i-2u\phi-\varphi_k\leq-p-1\\
-N-p\leq i-(2v+1)\phi-\varphi_j\leq-p-1
\end{array}\right.,
\end{equation}
where the lower and the upper bounds for $u$ and $v$ can be obtained as
\begin{equation}\label{equ11-cstc04}
\left\{
\begin{array}{l}u_1=\lceil\frac{i-\varphi_k+p+1}{2\phi}\rceil \\
v_1=\lceil\frac{i-\varphi_j+\phi+p+1}{2\phi}\rceil\\
u_2=\lfloor\frac{i-\varphi_k+N+p}{2\phi}\rfloor \\
v_2=\lfloor\frac{i-\varphi_j+\phi+N+p}{2\phi}\rfloor\\
 \end{array}\right.,
\end{equation}
where $\lceil\cdot\rceil$ and $\lfloor\cdot\rfloor$ denote ceil and floor functions, respectively.

It is reasonable to assume that $N>2\phi$ since the number of the data symbols in one frame should not be too small in order to achieve a
reasonable spectrum efficiency. Because $N>2\phi$ and recalling the assumption $\phi\geq \max\{\varphi_k,\varphi_j\}+1$ in Section \ref{sec2}, we can obtain that $0\leq v_1\leq v_2$ and $0\leq u_1\leq u_2$ in (\ref{equ11-cstc04}), that is, there are always non-negative solutions to (\ref{equ11-cstc03})
for $u$ and $v$.
Note that, for any $i$,
when $u_1=0$ or $v_1=0$, the corresponding term
in the  summation in the first line of (\ref{equ11-cstc02}) becomes zero.

Thus, the first line of (\ref{equ11-cstc02}) is made up of the symbols of
the $0$th frame and the zero symbols of the $-1$th frame (if any). The second line involves the non-zero symbols in the $-1$th frame. In fact,
non-zero symbols in other frames (e.g., $-2$th, $-3$th,  ...) are also similarly involved in the third line of (\ref{equ11-cstc02}).
All the non-zero symbols in the previous frames ($-1$th, $-2$th, $-3$th,  ...) are interference for $r^{(k)}(i)$ of the $0$th frame. Thus,
it shows that this type of inter-frame interference cannot be avoided no matter how many zeros are padded in $x(l)$, i.e., no matter how large $p$ is in (\ref{equ11-s0}).

\setcounter{equation}{24}
\begin{figure*}[!t]
\normalsize
\setcounter{MYtempeqncnt}{\value{equation}}
\begin{equation}\label{equ11-1}
\begin{array}{rcl}
\tilde{t}^{(2)}(i)\!=\left\{\!\!\!\!\!\!\!\!\begin{array}{lr}\begin{array}{rcl}&&\!\!\!\!\sum\limits_{n=0}^{\frac{L-3}{2}}\eta^n\left[\beta_2 h_{SR}^{(2)}{x}(i-(2n+1)\phi-\varphi_2)+\!\!\beta_2 h_{12}\beta_1h_{SR}^{(1)}{x}(i-(2n+2)\phi-\varphi_1)\right] \\
&&+\sum\limits_{n=0}^{\frac{L-3}{2}}\eta^n\left[\beta_2{w}^{(2)}(i-(2n+1)\phi)+\!\!\beta_2 h_{12}\beta_1{w}^{(1)}(i-(2n+2)\phi)\right]\\
&&+\eta^{\frac{L-1}{2}}\beta_2h_{SR}^{(2)}{x}(i-L\phi-\varphi_2)+\eta^{\frac{L-1}{2}}\beta_2h_{SR}^{(2)}{w}^{(2)}(i-L\phi)\\
&&+\,\,\eta^{\frac{L-1}{2}}\beta_2h_{12}{\tilde{t}}^{(1)}(i-L\phi),
\end{array},\,\,\text{$L$ is odd}\\
\begin{array}{rcl}&&\!\!\!\!\sum\limits_{n=0}^{\frac{L}{2}-1}\eta^n\left[\beta_2 h_{SR}^{(2)}{x}(i-(2n+1)\phi-\varphi_2)+\!\!\beta_2 h_{12}\beta_1h_{SR}^{(1)}{x}(i-(2n+2)\phi-\varphi_1)\right] \\
&&+\sum\limits_{n=0}^{\frac{L}{2}-1}\eta^n\left[\beta_2{w}^{(2)}(i-(2n+1)\phi)+\!\!\beta_2 h_{12}\beta_1{w}^{(1)}(i-(2n+2)\phi)\right]\\
&&+\,\,\eta^{\frac{L}{2}}{\tilde{t}}^{(2)}(i-L\phi),
\end{array},\,\,\text{$L$ is even}
\end{array}\right.
\end{array}
\end{equation}
\begin{equation}\label{equ11-101}
\begin{array}{rcl}
\tilde{t}^{(2)}(i)\!=\left\{\!\!\!\!\!\!\!\!\begin{array}{lr}\begin{array}{rcl}&&\!\!\!\!\sum\limits_{n=0}^{\frac{L-3}{2}}\eta^n\left[\beta_2 h_{SR}^{(2)}{x}(i-(2n+1)\phi-\varphi_2)+\!\!\beta_2 h_{12}\beta_1h_{SR}^{(1)}{x}(i-(2n+2)\phi-\varphi_1)\right] \\
&&+\sum\limits_{n=0}^{\frac{L-3}{2}}\eta^n\left[\beta_2{w}^{(2)}(i-(2n+1)\phi)+\!\!\beta_2 h_{12}\beta_1{w}^{(1)}(i-(2n+2)\phi)\right]\\
&&+\eta^{\frac{L-1}{2}}\beta_2h_{SR}^{(2)}{x}(i-L\phi-\varphi_2)+\eta^{\frac{L-1}{2}}\beta_2h_{SR}^{(2)}{w}^{(2)}(i-L\phi),
\end{array},\,\,\text{$L$ is odd}\\
\begin{array}{rcl}&&\!\!\!\!\sum\limits_{n=0}^{\frac{L}{2}-1}\eta^n\left[\beta_2 h_{SR}^{(2)}{x}(i-(2n+1)\phi-\varphi_2)+\!\!\beta_2 h_{12}\beta_1h_{SR}^{(1)}{x}(i-(2n+2)\phi-\varphi_1)\right] \\
&&+\sum\limits_{n=0}^{\frac{L}{2}-1}\eta^n\left[\beta_2{w}^{(2)}(i-(2n+1)\phi)+\!\!\beta_2 h_{12}\beta_1{w}^{(1)}(i-(2n+2)\phi)\right],
\end{array},\,\,\text{$L$ is even}
\end{array}\right.
\end{array}
\end{equation}

\setcounter{equation}{\value{MYtempeqncnt}}
\hrulefill
\end{figure*}
\setcounter{equation}{14}

To avoid the inter-frame interference in $r^{(k)}(i)$,
another zero padding at the relay transmission is needed,
which means that the relays send nothing but just keep receiving
during the zero padding period. The truly transmitted signals at the
relays during the $0$th frame have the following form:
\begin{equation}\label{equ11-a01}
\tilde{t}^{(k)}(i)=\left\{\begin{array}{cc}
0,&0 \leq i\leq \phi-1 \\
\beta_k \tilde{r}^{(k)}(i-\phi),&\phi \leq i< N+p
\end{array}\right.,
\end{equation}
The signal received at Relay $k$ after the self-loop interference signal is removed is
\begin{equation}\label{equ11-r01}
\tilde{r}^{(k)}(i)=h_{SR}^{(k)}x(i-\varphi_k)+h_{jk}\tilde{t}^{(j)}(i)+w^{(k)}_R(i), 0\leq i< N+p,
\end{equation}
where $k$ and $j$ are the indexes of the two relays.
Let us see what the signals sent in (\ref{equ11-a01})
at the relays are in details.
The signal sent by Relay 1 can be written as follows in two segments.

When $0 \leq i<\phi$, Relay 1 sends nothing and just keeps receiving, that is,
\begin{equation}\label{equ11-ac02}
\tilde{t}^{(1)}(i)=0.
\end{equation}

When $\phi \leq i< N+p$, the signal sent by Relay 1 can be written as (\ref{equ11-a02}),
where $L=\lfloor \frac{i}{\phi}\rfloor$ and $\eta=\beta_1\beta_2h_{12}h_{21}$. 
\setcounter{equation}{19}
It is clear that $0\leq i-L\phi \leq \phi-1$. This implies
$\tilde{t}^{(k)}(i-L\phi)=0$ for $k=1,2$ from (\ref{equ11-a01}).
Therefore, (\ref{equ11-a02}) can be rewritten as (\ref{equ11-a02-2}).


The signal part involving $x(i)$  in (\ref{equ11-a02-2}) can be re-formulated as follows:
\begin{equation}
 \mathbf{t}^{(1)}=[\tilde{t}^{(1)}(0),\tilde{t}^{(1)}(1),\cdots,\tilde{t}^{(1)}(N+p-1)].
\end{equation}
Then,  it can be regarded as a part of the following convolution
 \begin{equation}
 \mathbf{\tilde{m}}_1*\mathbf{x}=[\mathbf{0}_{\phi+\varphi_1} \,\,\mathbf{m}_1\,\,\mathbf{0}_{\max\{\varphi_1,\varphi_2\}-\varphi_1}]*\mathbf{x},
\end{equation}
where $\mathbf{x}=[x(0),x(1),\ldots,x(N+p-1)]=[s(0),s(1),\ldots,s(N-1), \mathbf{0}_p]$, $\mathbf{0}_k$
is the all zero vector of size $k$,
and
\begin{equation}\label{equ_m1}
\begin{array}{ll}\mathbf{m}_1\!\!=\![\!\!\!\!\!&\beta_1 h_{SR}^{(1)} \,\,\,\mathbf{0}_{\phi+\varphi_2-\varphi_1-1}\,\,\,\beta_1 h_{21}\beta_2h_{SR}^{(2)}\,\,\, \mathbf{0}_{\phi+\varphi_1-\varphi_2-1} \,\cdots\\
&\eta^{\Gamma }\beta_1 h_{SR}^{(1)}\,\, \mathbf{0}_{\phi+\varphi_2-\varphi_1-1} \,\,\eta^{\Gamma }\beta_1 h_{21}\beta_2h_{SR}^{(2)}\,\, \mathbf{0}_{\phi+\varphi_1-\varphi_2-1}]\end{array}\!\!,
\end{equation}
where $\Gamma=\lfloor \frac{N+p-1-\phi}{2\phi}\rfloor$.
In fact, we have
\begin{equation}\label{23}
\mathbf{t}^{(1)}=[\mathbf{\tilde{m}}_1*\mathbf{x}]_{ N+p},
\end{equation}
where $[\mathbf{u}]_{k}$ denotes the vector formed by the first $k$ elements of $\mathbf{u}$, i.e.,
$[\mathbf{u}]_{k}=[u(0), u(1), \cdots, u(k-1)]$.
From (\ref{23}), one can see that the effective signal part $\mathbf{t}^{(1)}$
transmitted at Relay 1 is a coded signal of the original
data sequence $s(i)$ with the generator sequence $\mathbf{\tilde{m}}_1$.

The same as Relay 1, the effective signal part sent from Relay 2 is as follow.

When $0\leq i< \phi$, Relay 2 transmits nothing but keeps receiving, that is,
\begin{equation}\label{equ11-ac03}
\tilde t^{(2)}(i)=0.
\end{equation}

When $\phi \leq i< N+p-1$, the signal sent by Relay 2 can be written as (\ref{equ11-1}).
From (\ref{equ11-a01}), it is clear that $0 \leq i-L\phi\leq \phi-1$, ${t}^{(k)}(i-L\phi)=0$ for $k=1,2$. Thus, (\ref{equ11-1}) can also be rewritten as (\ref{equ11-101}).
\setcounter{equation}{26}
The signal part involving $x(i)$  in (\ref{equ11-101}) can be re-formulated as follows:
\begin{equation}
\mathbf{t}^{(2)}=[\tilde{t}^{(2)}(0),\tilde{t}^{(2)}(1),\cdots,\tilde{t}^{(2)}(N+p-1)],
\end{equation}
which can also be written as
\begin{equation}\label{30}
\mathbf{t}^{(2)}=[\mathbf{\tilde{m}}_2*\mathbf{x}]_{ N+p},
\end{equation}
where
\begin{equation}
\mathbf{\tilde{m}}_2=[\mathbf{0}_{\phi+\varphi_2}  \,\,\mathbf{m}_2\,\,\mathbf{0}_{\max\{\varphi_1,\varphi_2\}-\varphi_2}],
\end{equation}
and
\begin{equation}\label{equ_m2}
\begin{array}{ll}\mathbf{m}_2\!\!=\![\!\!\!\!\!&\beta_2 h_{SR}^{(2)} \,\,\,\mathbf{0}_{\phi+\varphi_1-\varphi_2-1} \,\,\,\beta_2 h_{12}\beta_1h_{SR}^{(1)}\,\,\, \mathbf{0}_{\phi+\varphi_2-\varphi_1-1}  \cdots\\
&\eta^{\Gamma}\beta_2 h_{SR}^{(2)}\,\, \mathbf{0}_{\phi+\varphi_1-\varphi_2-1} \,\,\eta^{\Gamma}\beta_2 h_{12}\beta_1h_{SR}^{(1)}\,\,\mathbf{0}_{\phi+\varphi_2-\varphi_1-1}].
\end{array}
\end{equation}
Thus, the effective signal part $\mathbf{t}^{(2)}$
transmitted at Relay 2 is also a coded signal of the original
data sequence $s(i)$ with the generator sequence $\mathbf{\tilde{m}}_2$.

To normalize the mean transmission power at the relays, the amplifying factors should satisfy
\begin{equation}\label{equ13-2}
E[|\mathbf{\tilde{m}}_1|^2]=E[|\mathbf{\tilde{m}}_2|^2]=1.
\end{equation}
Substituting $\mathbf{\tilde{m}}_1$ and $\mathbf{\tilde{m}}_2$ into (\ref{equ13-2}) and considering $E[|h_{SR}^{(1)}|^2]=E[|h_{SR}^{(2)}|^2]=1$, we obtain
\begin{equation}\label{equ14-0}
\left\{\begin{array}{lcc}
\sum_{n=0}^{\Gamma}|\eta|^{2n} \left(|\beta_1|^2+|\beta_1 h_{21}\beta_2|^2\right)&=& 1 \\
 \sum_{n=0}^{\Gamma}|\eta|^{2n} \left(|\beta_2|^2+|\beta_2 h_{12}\beta_1|^2\right)&=& 1
\end{array}\right.,
\end{equation}
where the channel coefficients $h_{jk}$ between the two relays
are treated deterministic and known at the relays in the above equation.
This is because $\eta=\beta_1\beta_2 h_{12} h_{21}$ and if we treat $h_{jk}$
as random variables similar
to $h_{SR}^{(j)}$, it will be not possible to have their moments of
all even orders needed in (\ref{equ14-0}) and thus not possible to solve
$\beta_j$.
From (\ref{equ14-0}),
the amplifying factors $\beta_1$ and $\beta_2$ can be found.
\setcounter{equation}{35}
\begin{figure*}[!t]
\normalsize
\setcounter{MYtempeqncnt}{\value{equation}}
\begin{equation}\label{codeMatrix001}
\mathbf{\hat{A}}=\left[\begin{array}{cc}\mathbf{\hat{a}}_1\\
\mathbf{\hat{a}}_2\end{array}\right]
=\left[\begin{array}{c} \left[a_{11} \,\, \mathbf{0}_{\psi_1} \,\, a_{12} \,\, \mathbf{0}_{\xi-\psi_1}\,\, Ka_{11} \,\, \mathbf{0}_{\psi_1} \,\, Ka_{12} \,\, \mathbf{0}_{\xi-\psi_1}\,\, K^2a_{11} \,\, \mathbf{0}_{\psi_1} \,\, K^2a_{12} \,\, \mathbf{0}_{\xi-\psi_1} \cdots\right]\\
\left[a_{21} \,\, \mathbf{0}_{\psi_2}\,\, a_{22} \,\, \mathbf{0}_{\xi-\psi_2}\,\,Ka_{21} \,\, \mathbf{0}_{\psi_2}\,\, Ka_{22} \,\, \mathbf{0}_{\xi-\psi_2}\,\,K^2a_{21} \,\, \mathbf{0}_{\psi_2}\,\, K^2a_{22} \,\, \mathbf{0}_{\xi-\psi_2} \cdots\right]\end{array}\right]
\end{equation}

\setcounter{equation}{\value{MYtempeqncnt}}
\hrulefill
\end{figure*}

\setcounter{equation}{32}
From (\ref{equ11-a02-2}), (\ref{equ11-101}) and (\ref{equ14-0}), it is not hard to see that the noise in the signal to be sent at the relays is zero mean and with the variance of $\sigma_R^2$. So the transmission SNRs at the two relays are
\begin{equation}\label{equ14-1}
\gamma_1=\gamma_2=\frac{1}{\sigma_R^2}.
\end{equation}
We can see that unlike (\ref{equ2-14}), the noises are not amplified or accumulated at the relay nodes.


From (\ref{23}) and (\ref{30}), one can see that the two transmitted signals at the two relays are generated by the same signal ${\bf x}$ with two generator
sequences $\tilde{\bf m}_j$, $j=1,2$. If these two generator sequences
are put to a generator matrix of two rows as

\begin{equation}\label{equ_orig_mat1}
\mathbf{\tilde{M}}=\left[\begin{array}{c} \mathbf{\tilde{m}}_1\\
\mathbf{\tilde{m}}_2\end{array}\right]=\left[\begin{array}{c} \mathbf{0}_{\phi+\varphi_1}  \,\,\mathbf{m}_1\,\,\mathbf{0}_{\max\{\varphi_1,\varphi_2\}-\varphi_1}\\
\mathbf{0}_{\phi+\varphi_2}  \,\,\mathbf{m}_2\,\,\mathbf{0}_{\max\{\varphi_1,\varphi_2\}-\varphi_2}\end{array}\right],
\end{equation}
then, the two transmitted signals at the two relays are the outputs
of the signal ${\bf x}$ with the above encoding generator matrix.
At the destination, these two signals are received through two fading channels $h_{RD}^{(j)}$, $j=1, 2$.
What we are interested now is whether this system can achieve the spatial diversity of two
from the two relays. Similar to the case studied in \cite{Yiliu},
these two signals  may not be synchronized and
may arrive at the destination at different times. Thus, we may apply the theory of shift full rank (SFR)
matrices developed in  \cite{4155128, Shang001, Guo} and,
 to check the diversity property, we need to check whether the generator matrix $\tilde{\bf M}$
is SFR, where the synchronization between the two relays may not be achieved.

Different from \cite{Yiliu} where the coding process at the relay is independent of the other link, i.e., the direct link, the coding processes at the two relays are not independent in this case because the data are exchanged through the cross-talk links. For this reason, the source to relay channels $h_{SR}^{(j)}$, $j=1, 2$, have nonlinear effect on the generator matrix $\tilde{\bf M}$, which can be seen from (\ref{equ_m1}) and (\ref{equ_m2}). So $h_{SR}^{(j)}$, $j=1, 2$, are included in $\tilde{\bf M}$ when checking its SFR property in the following.

\subsection{Diversity analysis}\label{sec3b}

An SFR matrix is a matrix that has full row rank no matter how its rows are shifted.
For more about SFR matrices, we refer to \cite{4155128, Shang001, Guo}.
Clearly, the SFR property of the generator matrix $\tilde{\bf M}$ is equivalent to
that of the following matrix
\begin{equation}\label{equ12}
\mathbf{\hat{M}}=\left[\begin{array}{c} \mathbf{{m}}_1\\
\mathbf{{m}}_2\end{array}\right],
\end{equation}
which will be studied in this subsection.
To do so,
we first have the  following lemma.
\begin{Lemma}\label{lemma1}

Let $K$ be a constant. Matrix $\mathbf{\hat{A}}$ defined by (\ref{codeMatrix001})
is an SFR matrix for any $\psi_1$ and $\psi_2$ with $0\leq \psi_1, \psi_2 \leq \xi$ if and only if $\mathbf{A}=\left[\begin{array}{cc} a_{11}&a_{12}\\
a_{21}&a_{22}\end{array}\right]$ is an SFR matrix.
\end{Lemma}
\begin{proof}
In \cite{Guo}, the necessary and sufficient condition for two-row SFR matrix is: {\it The two-row matrix is an SFR matrix iff the two rows are linearly independent.}

It is not hard to prove that the two rows of $\mathbf{\hat{A}}$ and the two rows of $\mathbf{A}$ are either both linearly independent or both linearly dependent. So we obtain that $\mathbf{\hat{A}}$ and $\mathbf{A}$ have the same SFR property.
%
\end{proof}
\setcounter{equation}{36}
\begin{Lemma}\label{lemma1-0}
Suppose $\mathbf{\hat{A}}=\left[\begin{array}{cc}\mathbf{\hat{a}}_1\\
\mathbf{\hat{a}}_2\end{array}\right]$ is defined by (\ref{codeMatrix001})
and $\mathbf{x}$ is $1\times (N+p)$ vector. If and only if $p\geq \xi+1$, the partial DLC-STC
\begin{equation}\label{equ-code}
\mathbf{\hat{C}}=\left[\begin{array}{cc}\left[\mathbf{\hat{a}_1}*\mathbf{x}\right]_{N+p}\\
\left[\mathbf{\hat{a}_2}*\mathbf{x}\right]_{N+p}\end{array}\right]
\end{equation}
can achieve the same asynchronous diversity as the DLC-STC
\begin{equation}\label{equ-dlc-stc}
\mathbf{\bar{C}}=\left[\begin{array}{cc}\bar{\mathbf{a}}_1*\mathbf{x}\\
\bar{\mathbf{a}}_2*\mathbf{x}\end{array}\right]
\end{equation}
Here, $\bar{\mathbf{a}}_1=\left[a_{11} \,\,\,\, \mathbf{0}_{\psi_1} \,\,\,\, a_{12} \,\,\,\, \mathbf{0}_{\xi-\psi_1}\right]$ and $\bar{\mathbf{a}}_2=\left[a_{21} \,\,\,\, \mathbf{0}_{\psi_2}\,\,\,\, a_{22} \,\,\,\, \mathbf{0}_{\xi-\psi_2}\right]$, where $\psi_1$, $\psi_2$, and $K$ are constants with $0\leq \psi_1, \psi_2 \leq \xi$.
\end{Lemma}

The proof is in Appendix A.

For the model considered in this paper, $\mathbf{\hat{A}}$ in (\ref{codeMatrix001}) equals $\mathbf{\hat{M}}$ in (\ref{equ12}) and $\mathbf{\bar{A}}=\left[\begin{array}{c} \mathbf{\bar{a}}_1\\ \mathbf{\bar{a}}_2\end{array}\right]$ in (\ref{equ-dlc-stc}) equals $\mathbf{\bar{M}}=\left[\begin{array}{c} \mathbf{\bar{m}}_1\\ \mathbf{\bar{m}}_2\end{array}\right]=\left[\begin{array}{cccc} \beta_1 h_{SR}^{(1)}&\mathbf{0}_{\phi+\varphi_2-\varphi_1-1}&\beta_1 h_{21}\beta_2h_{SR}^{(2)}&\mathbf{0}_{\phi+\varphi_1-\varphi_2-1}\\ \beta_2 h_{SR}^{(2)} &\mathbf{0}_{\phi+\varphi_1-\varphi_2-1}&\beta_2 h_{12}\beta_1h_{SR}^{(1)}&\mathbf{0}_{\phi+\varphi_2-\varphi_1-1}\end{array}\right]$. And $\psi_1=\phi+\varphi_2-\varphi_1-1$, $\psi_2=\phi+\varphi_1-\varphi_2-1$, $\xi=2\phi-2$. Thus, $p\geq2\phi-1$. The two transmitted signals $\mathbf{t}^{(1)}$ in (\ref{23}) and $\mathbf{t}^{(2)}$ in (\ref{30}) are truncated sequences of $\mathbf{\tilde{m}}_1*\mathbf{x}$ and $\mathbf{\tilde{m}}_2*\mathbf{x}$,
but when $p\geq 2\phi-1$, they contain $\mathbf{\bar{t}}_1=\mathbf{\bar{m}}_1*\mathbf{x}$ and $\mathbf{\bar{t}}_2=\mathbf{\bar{m}}_2*\mathbf{x}$, respectively. From Lemma 2, we only need to check the SFR property of $\mathbf{\bar{M}}$.


As for the SFR property of the generator matrix of the proposed scheme, we have the following theorem.
\begin{Theorem}\label{Theorem1}
$\mathbf{\hat{M}}$ in (\ref{equ12}) (or $\mathbf{\bar{M}}$) is SFR if and only if $\beta_1h_{12}(h_{SR}^{(1)})^2\neq \beta_2h_{21}(h_{SR}^{(2)})^2$, $\beta_1 h_{SR}^{(1)}\neq0$, $\beta_2 h_{SR}^{(2)}\neq0$, where $\beta_1$ and $\beta_2$ are defined by (\ref{equ14-0}).
\end{Theorem}
\begin{proof}
From Lemma~\ref{lemma1}, we only need to investigate the SFR property of
\begin{equation}
\mathbf{M}=
\left[\begin{array}{cc}
\beta_1 h_{SR}^{(1)} &\beta_1 h_{21}\beta_2h_{SR}^{(2)} \\
\beta_2 h_{SR}^{(2)} &\beta_2 h_{12}\beta_1h_{SR}^{(1)}
\end{array}\right].
\end{equation}

First, we prove the necessary condition. If $\beta_1 h_{SR}^{(1)}=0$ or $\beta_2 h_{SR}^{(2)}=0$, it is easy to see that $\mathbf{M}$ is not an SFR matrix. If $\beta_1h_{12}(h_{SR}^{(1)})^2= \beta_2h_{21}(h_{SR}^{(2)})^2$, we obtain $\frac{\beta_1 h_{SR}^{(1)}}{\beta_2 h_{SR}^{(2)}}=\frac{\beta_1 h_{21}\beta_2h_{SR}^{(2)}}{\beta_2 h_{12}\beta_1h_{SR}^{(1)}}$, which means the two rows of $\mathbf{M}$ are not linearly independent. So the two-row matrix $\mathbf{M}$ is not an SFR matrix.

Second, we prove the sufficient condition. If $\beta_1h_{12}(h_{SR}^{(1)})^2\neq \beta_2h_{21}(h_{SR}^{(2)})^2$ and $\beta_1 h_{SR}^{(1)}\neq0,\beta_2 h_{SR}^{(2)}\neq0$, we can obtain that $\frac{\beta_1 h_{SR}^{(1)}}{\beta_2 h_{SR}^{(2)}}\neq \frac{\beta_1 h_{21}\beta_2h_{SR}^{(2)}}{\beta_2 h_{12}\beta_1h_{SR}^{(1)}}$, which means the two rows of $\mathbf{M}$ are linearly independent. So the two-row matrix $\mathbf{M}$ is an SFR matrix.
\end{proof}




It is proved in \cite{4155128, Shang001, Guo} that in a half-duplex relay system, the SFR property of a generator matrix leads to the asynchronous full diversity of the DLC-STC, where the coefficients (or entries) of the generator matrix are fixed constants with a fixed total power. It is, however, different, in the above full duplex system where the entries of the generator matrix in Theorem \ref{Theorem1} depend on the source to relay channel coefficients that are random. Although this is the case, due to the amplifying factor $\beta_k$ used at the relay nodes, the total power of the generator matrix coefficients is also fixed. As a result, the full asynchronous diversity is achieved from our numerous simulations due to the above SFR property.

\setcounter{equation}{42}
\begin{figure*}[!t]
\normalsize
\setcounter{MYtempeqncnt}{\value{equation}}
\begin{equation}\label{equ11-00}
t(b\phi+\varphi+l)=\beta \sum\limits_{j=1}^{b}(h_{LI}\beta)^{j-1}[h_{SR}x((b-j)\phi+l)+w_R((b-j)\phi+\varphi+l)]
\end{equation}
\begin{equation}\label{equ11v2}
\begin{array}{lll}
r(b\phi+\varphi+l)&\!\!=h_{SR}x(b\phi+l)+h_{LI}t(b\phi+\varphi+l)+w_R(b\phi+\varphi+l)\\
      &=\sum\limits_{j=1}^{b}(h_{LI}\beta)^{j-1}h_{SR}x((b-j+1)\phi+l)+(h_{LI}\beta)^{b}[h_{SR}x(l)+w_R(\varphi+l)]\\
      &+\sum\limits_{j=1}^{b}(h_{LI}\beta)^{j-1}w_R((b-j+1)\phi+\varphi+l)
\end{array}
\end{equation}
\setcounter{equation}{45}
\begin{equation}\label{equ11-21}
\begin{array}{lll}
t((b+k)\phi+\varphi+l)\!\!&=\beta [r((b+k-1)\phi+\varphi+l)-(h_{LI}\beta)^{b}\hat{x}((k-1)\phi+l)]\\
&=\beta\sum\limits_{j=1}^{b}(h_{LI}\beta)^{j-1}h_{SR}x((b+k-j)\phi+l)\\
&+ \beta\sum\limits_{j=1}^{b}(h_{LI}\beta)^{j-1}w_R((b+k-j)\phi+\varphi+l).
\end{array}
\end{equation}
\setcounter{equation}{\value{MYtempeqncnt}}
\hrulefill
\end{figure*}
\section{Coding for the case without cross-talk}\label{sec5-1}
In some scenarios, the signals from the cross-talk channels are too weak to consider, e.g., when directional antennas are used at the relay transmitter. In this case, the system model becomes what is shown in Fig.~\ref{fig_no_cross-talk} \cite{5161790, 5089955,5961159,4917875}. In this section, we will present a DLC-STC scheme for this case.

\begin{figure}[htbp]
\centering
\includegraphics[scale=0.8]{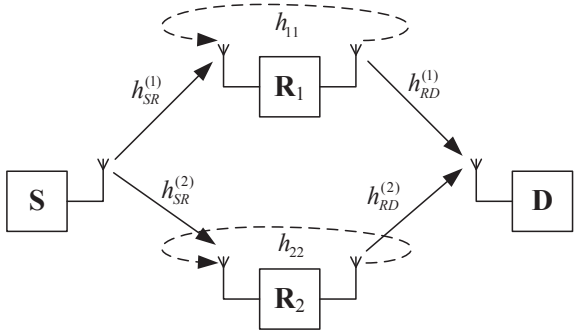}
\caption{\label{fig_no_cross-talk}Two-relay two-hop cooperative network without cross-talk.}
\end{figure}

\subsection{Construction of DLC-STC}\label{sub4a}
Because no cross-talk is considered in this case, for which the two relays are independent and one of the relays can be taken as an example to describe the coding process, we omit the relay index and denote the loop channel response as $h_{LI}$ during the following discussions for simplicity. The received signal $r(i)$ and the transmitted signal $t(i)$ at the relay node at time $i$ are as follows \cite{5089955,5961159,4917875},
\setcounter{equation}{39}
\begin{eqnarray}
r(i)&=&h_{SR}x(i-\varphi)+ h_{LI}t(i)+w_R(i)\label{equ101}\\
t(i)&=&\beta r(i-\phi)\label{equ201},
\end{eqnarray}
where, note that, $\varphi$ is the delay from the source to the relay, and $\phi$ is the common delay at the relays.

In this case, we can see that no cross-talk can be used to construct the code. However, we will present a DLC-STC scheme which is constructed by the signals from loop channels. Following the same idea as Scheme 2 in \cite{Yiliu}, where the DLC-STC is designed for the network with only one relay and the direct link, this scheme is to cancel the loop channel partially and do the coding by making use of the signal from the loop channel at the relay.

Since the loop channel information is known at the relay, the signal from the loop channel can be removed and the signal from the source node can be estimated at the relay itself at time $i$ as

\begin{equation}\label{equ701}
\hat{x}(i-\varphi)=r(i)-h_{LI}t(i)=h_{SR}x(i-\varphi)+w_R(i),
\end{equation}
which will be used in the residual interference cancellation later.

In this case, the main idea is to construct the convolutional code at each relay by the feedback of the loop interference channel. Next, we will see the process in details.

Suppose $b(\geq2)$ symbols are to be coded. In (\ref{equ201}), by letting $i=b\phi+\varphi+l, 0\leq l\leq \phi-1$ and considering $\{x(i)=0,w_R(i)=0,i\leq -1\}$, we can obtain the transmission signal at the relay at time slot $b\phi+\varphi+l$ as (\ref{equ11-00}).

Substituting (\ref{equ11-00}) into (\ref{equ101}), the received signal at the relay at time slot $b\phi+\varphi+l$ can be written as (\ref{equ11v2}).
The first term in the right hand side of the last equation in (\ref{equ11v2}) is the desired transmission signal including
$b$ symbols while the second and third terms are
interference and noise from the loop channel, respectively. Notice that from (\ref{equ701}), the second term in (\ref{equ11v2}) can be written as
$$(h_{LI}\beta)^{b}[h_{SR}x(l)+w_R(\varphi+l)]=(h_{LI}\beta)^{b}\hat{x}(l).$$
 At time slot $(b+1)\phi+\varphi+l$, since the relay has obtained
 the estimated signal $\hat{x}(l)$ as in (\ref{equ701}), the interference of the second term in (\ref{equ11v2})
can be cancelled. Then the transmission signal can be written as
\setcounter{equation}{44}
\begin{equation}\label{equ11-20}
\begin{array}{lll}
t((b+1)\phi\!+\!\varphi\!+\!l)\!\!\!\!&=\beta [r(b\phi+\varphi+l)-(h_{LI}\beta)^{b}\hat{x}(l)]\\
&=\beta\!\sum\limits_{j=1}^{b}(h_{LI}\beta)^{j-1}h_{SR}x((b\!-\!j\!+\!1)\phi\!+\!l)\\
&+\,\beta\!\sum\limits_{j=1}^{b}(h_{LI}\beta)^{j-1}w_R((b\!-\!j\!+\!1)\phi\!+\!\varphi\!+\!l).
\end{array}
\end{equation}
If the cancellation process is done continuously
in terms of $k$ in the time index $(b+k)\phi+\varphi+l$,
we can obtain a general expression as (\ref{equ11-21}).
Letting $i=(b+k)\phi+\varphi+l$, where $0\leq l\leq \phi-1$, (\ref{equ11-21}) can be simplified as
\setcounter{equation}{46}
\begin{equation}\label{equ11-2}
\begin{array}{rcl}
t(i)&=&\beta [r(i-\phi)-(h_{LI}\beta)^{b}\hat{x}(i-(b+1)\phi-\varphi)]\\
&=&\beta \sum\limits_{j=1}^{b}(h_{LI}\beta)^{j-1}h_{SR}x(i-j\phi-\varphi)\\
&&+\beta \sum\limits_{j=1}^{b}(h_{LI}\beta)^{j-1}w_R(i-j\phi).
\end{array}
\end{equation}
Let
\begin{equation}
q(i)=\left\{
\begin{array}{cc}
\beta(h_{LI}\beta)^{n-1},&i=n\phi+\varphi,  (1\leq n \leq b) \\
0,&else
\end{array}\right.,
\end{equation}
then the first
term in the right hand side of the last equation in (\ref{equ11-2}) can be written as the convolution
between $q(i)$ and $x(i)$ as follows,
\begin{equation}\label{equ13}
\beta\sum\limits_{j=1}^{b}(h_{LI}\beta)^{j-1}h_{SR}x(i-j\phi-\varphi)=h_{SR}\cdot q(i)\ast x(i),
\end{equation}
where $0<\beta<\frac{1}{|h_{LI}|}$ is the amplify parameter controlling the relay transmission power as
\begin{equation}\label{equ13-2v2}
E\left\{\sum_{i=1}^b|q(i)|^2\right\}=\sum_{i=1}^b|\beta(h_{LI}\beta)^{i-1}|^2=1.
\end{equation}

In (\ref{equ13-2v2}), $b$ determines the constraint length of the convolutional code.
To ensure full row rank of the effective coding matrix,
$b$ should be no less than the number of independent links, which is  $2$ in the current case. The sequence
$q(i)$ is determined after $\beta$ is selected to satisfy  (\ref{equ13-2v2}).

If we combine the two relay links and use the subscript indices to denote different relays,
the two signals transmitted from the two relays can be thought of as
a DLC-STC with the following generator matrix:
\begin{equation}\label{Mscheme2}
 \mathbf{M}=\left[\begin{array}{cccccc}
\beta_{1}&\!\!\mathbf{0}_{\phi-1}&\!\!\beta_{1}(h_{11}\beta_{1})&\!\!\mathbf{0}_{\phi-1}&\!\!\ldots&\!\!\beta_{1}(h_{11}\beta_{1})^{b-1}\\
 \beta_{2}&\!\!\mathbf{0}_{\phi-1}&\!\!\beta_{2}(h_{22}\beta_{2})&\!\!\mathbf{0}_{\phi-1}&\!\!\ldots&\!\!\beta_{2}(h_{22}\beta_{2})^{b-1}\end{array}\right],
\end{equation}
where $\beta_k,k=1,2$ are determined by
\begin{equation}\label{Mscheme3}
\left\{\begin{array}{c}
\sum\limits_{i=1}^b|\beta_1(h_{11}\beta_1)^{i-1}|^2=1\\
\sum\limits_{i=1}^b|\beta_2(h_{22}\beta_2)^{i-1}|^2=1\end{array}\right..
\end{equation}
The received signal at the destination is a superposition of possibly delayed versions of these two signals.

{\it Notice}: The basic idea of the scheme above is the same as Scheme 2 proposed in \cite{Yiliu}. However, Scheme 2 in \cite{Yiliu} is for the network with one relay and the direct link, where the delay between the source to the relay is counted in the total link delay and the processing delay at the relay is one symbol period. It can be regarded as one special case of the proposed scheme above by choosing $\phi=1$, $\varphi_1=\varphi_2=0$, $\beta_1=1$, and $h_{11}=0$.

\subsection{Diversity analysis}
In this section, we discuss the diversity of the proposed scheme. Similar to what was studied before, as what has been shown in \cite{4155128, Shang001, Guo}, to achieve the asynchronous full cooperative diversity, the SFR property of the generator matrix $\mathbf{M}$ in (\ref{Mscheme2}) plays the important role. From the result in \cite{Guo}, it it not hard to see that in (\ref{Mscheme2}), removing the zero columns doesn't change the SFR property of $\mathbf{M}$. So we can study the following generator matrix which is simpler:
\begin{equation}\label{Mscheme21}
 \bar{\mathbf{M}}=\left[\begin{array}{cccc}
\beta_{1}&\beta_{1}(h_{11}\beta_{1})&\ldots&\beta_{1}(h_{11}\beta_{1})^{b-1}\\
 \beta_{2}&\beta_{2}(h_{22}\beta_{2})&\ldots&\beta_{2}(h_{22}\beta_{2})^{b-1}\end{array}\right].
\end{equation}

\begin{Theorem}\label{Theorem3}
The generator matrix $\mathbf{M}$ in (\ref{Mscheme2}) is an SFR matrix iff $h_{11}\beta_1\neq h_{22}\beta_2$.
\end{Theorem}
\begin{proof}
By the necessary and sufficient condition for two-row SFR matrix given in \cite{Guo}, we obtain that $\bar{\mathbf{M}}$ is an SFR matrix iff the two rows $\bar{\mathbf{m}}_k=\left[\beta_k, \beta_k(h_{kk}\beta_k), \cdots, \beta_k(h_{kk}\beta_k )^{b-1}\right], k=1,2$, are linearly independent. It is obvious that $\beta_k\neq0, k=1,2$ since the transmission power at the relays is normalized to 1 in (\ref{Mscheme3}). Then it is easy to see that $\bar{\mathbf{m}}_1$ and $\bar{\mathbf{m}}_2$ are linearly independent iff $h_{11}\beta_1\neq h_{22}\beta_2$. Finally, we can obtain that $\bar{\mathbf{M}}$ or $\mathbf{M}$ is an SFR matrix iff $h_{11}\beta_1\neq h_{22}\beta_2$.
\end{proof}

By Theorem \ref{Theorem3}, we know that if the condition
$h_{11}\beta_1\neq h_{22}\beta_2$ is satisfied,
the designed DLC-STC in (\ref{equ11-2}) for the relays with full duplex loop channels is SFR. Although $h_{11}$ and $h_{22}$ are random, the total energy of the generator sequences in $\mathbf{M}$ are normalized by (\ref{Mscheme3}) and as a matter of fact, from our numerous simulations, the full diversity can be indeed achieved. Different from the partial DLC-STC scheme proposed in Section \ref{sec3}, the effect of the source to relay channels has the same linear function as that of the relay to destination channels, for which the effects of the source to relay channels and the relay to destination channels can be put together and regarded as an equivalent channel when analyzing the diversity.

\section{Simulations}\label{sec6}
In this section, we present some simulation results to illustrate the performance of the proposed DLC-STC schemes for two-relay full-duplex
cooperative networks. In simulations, all the wireless channels are set to be quasi-static Rayleigh flat fading. Because there are just two relays, we only need to consider the relative delay between the two relays. The delays, $\varphi_1$ and $\varphi_2$, from the source to the relays are chosen to be 0 or 1 with the same probabilities, where there are only three different cases as follows:
\begin{enumerate}
\item Case 1: Relay 1 is ahead of Relay 2 ($\varphi_1=0,\varphi_2=1$);
\item Case 2: Relay 2 is ahead of Relay 1 ($\varphi_1=1,\varphi_2=0$);
\item Case 3: no delay between Relay 1 and Relay 2 ($\varphi_1=\varphi_2=0$ or $\varphi_1=\varphi_2=1$).
\end{enumerate}
The delays from the relays to the destination are uniformly distributed in $[0,\tau_{max}-1]$.

The length of each symbol block $N=20$. The maximum delay $\tau_{max}$ is 3. The zero padding length is 6 in the proposed two full duplex schemes while it is 3 in other schemes. The common processing delay at the relays is $\phi=2$. The constellation used is QPSK. We compare the performance of the proposed two schemes (marked as FD cross-talk and FD no cross-talk, respectively), and the DLC-STC scheme for half duplex cooperative communications (marked as HD) \cite{Guo}. We also simulate the proposed space-time code constructed using the loop-interference as self-coding in \cite{Yiliu} for full-duplex cooperative networks with one relay and the direct link from the source to the destination nodes for comparison (marked as Self-coding). In the Self-coding scheme, the direct link channel is also modeled as quasi-static flat Rayleigh fading channel and the direct link is regarded as a special relay whose generator polynomial is defined by the vector $[1 \;\; \mathbf{0}_{(b-1)\phi}]$. In the HD scheme and FD no cross-talk scheme, the effective symbols to be convoluted in one time slot is $b=3$. With the parameters above, the spectrum efficiency of each scheme is given in the legend of each figure. All the schemes are evaluated with MMSE-DFE receivers \cite{869048,5683902}. The signal to noise ratios (SNRs) at the receivers of the relay and the destination are denoted as $S\!N\!R_R$ and $S\!N\!R_D$, respectively. Since the average power gain of each wireless Rayleigh flat fading channel is normalized to be 1, we have $S\!N\!R_R=\frac{E_s}{\sigma^2_R}$ and $S\!N\!R_D=\frac{E_s}{\sigma^2_D}$.

{\it Simulation 1-BER vs. $S\!N\!R_D$}:
In this simulation, we compare the BER performance vs. $S\!N\!R_D$ when $S\!N\!R_R$ remains constant as 30dB, which is shown in Fig.~\ref{fig4}. We can see that the BER performance of Self-coding scheme is the best. This is reasonable since the direct link is available in this scheme and there is no source to relay fading or relay noise in the direct link. As for the schemes for FD mode, FD cross-talk and FD no cross-talk almost have the same BER performance. The difference of the BER performance between the two FD schemes and HD scheme is not very large, but the two FD schemes can achieve a much higher throughput than HD scheme.
\begin{figure}[htbp]
\centering
\includegraphics[scale=0.7]{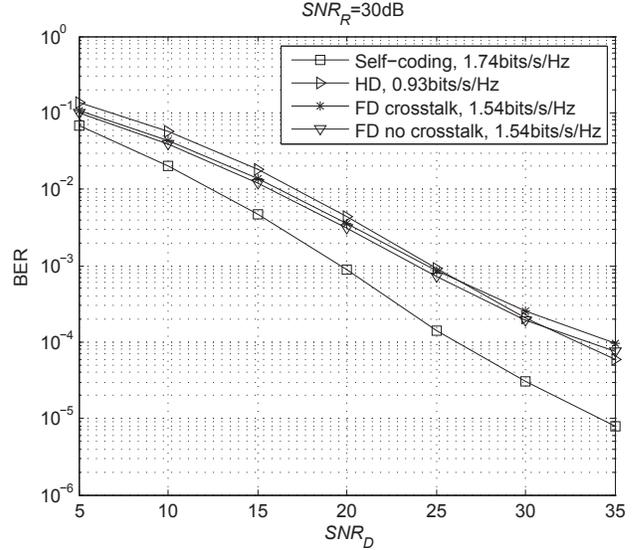}
\caption{\label{fig4} BER versus $S\!N\!R_D$ of two-relay full-duplex network with AF protocol when $S\!N\!R_R$ remains constant as 30dB.}
\end{figure}

{\it Simulation 2-BER vs. $S\!N\!R_R$}:
In this simulation, the BER performance vs. $S\!N\!R_R$ is compared when $S\!N\!R_D$ remains constant as 30dB, which is shown in Fig.~\ref{fig5}. We also notice that the BER performance of Self-coding scheme is the best. As for the two schemes for the FD mode, the BER performance is almost the same. The BER performance of HD scheme is a little better than those of the two FD schemes but its throughput is much smaller. From both the figures, we can see that $S\!N\!R_D$ has much more effect on HD scheme than the two FD schemes.
\begin{figure}[htbp]
\centering
\includegraphics[scale=0.7]{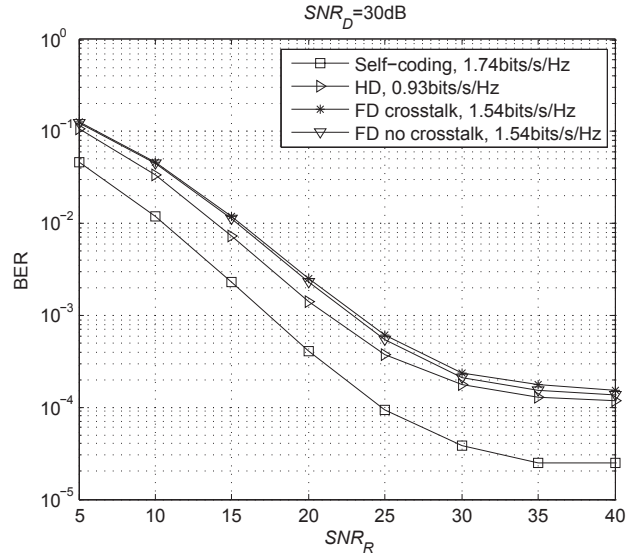}
\caption{\label{fig5} BER versus $S\!N\!R_R$ of two-relay full-duplex network with AF protocol when $S\!N\!R_D$ remains constant as 30dB.}
\end{figure}

{\it Simulation 3-Diversity comparison}: This simulation is to compare the achievable diversity of different schemes. In this simulation, we set $S\!N\!R_R=S\!N\!R_D=\gamma$. The receivers for all the schemes are MMSE-DFE receivers. The direct transmission scheme (marked as Direct transmission), in which the source node sends the signal to the destination node directly without any relay, is also simulated for comparison. To be fair, the transmission power for direct transmission scheme is doubled, so the SNR is $\frac{2Es}{\sigma_D^2}=2{\gamma}$. We can see that except the Direct transmission scheme which has only one transmit antenna, all other schemes who have two equivalent transmit antennas can achieve the full diversity, that is, diversity order of two, when $\gamma$ goes to infinity.
\begin{figure}[htbp]
\centering
\includegraphics[scale=0.7]{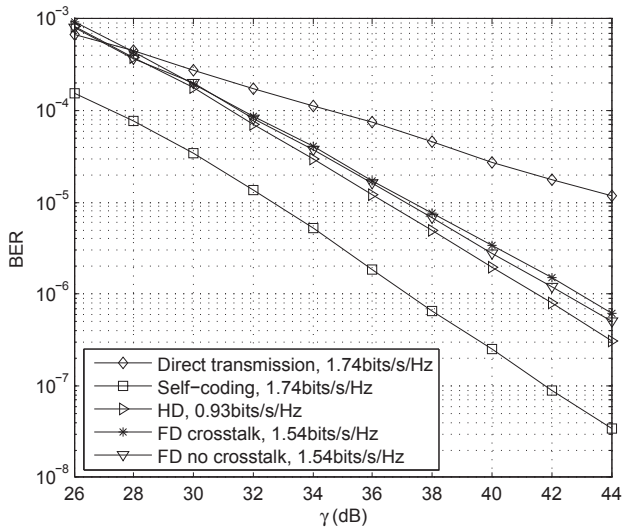}
\caption{\label{fig6} Diversity comparison of two-relay full-duplex network with AF protocol when SNR is high.}
\end{figure}

{\it Simulation 4-Diversity with direct link}: Although the direct link is not considered during the analysis, it is straightforward to include the direct link by just adding one more row for the generator matrix representing ``a special relay" whose generator polynomial is $[1\,\,\, 0\,\,\, 0\,\, \cdots\,\, 0]$. Moreover, the 3 by 1 MIMO with the linear Toeplitz space-time code \cite{1523684} is also simulated for comparison. To be fair, we set the 3 by 1 MIMO to have the same received SNR with the proposed schemes, that is, if $S\!N\!R_R=S\!N\!R_D=\gamma$ dB, the received SNR for 3 by 1 MIMO is $\gamma+10log_{10}3$ dB. The Direct transmission without any relay scheme, which has the diversity order of one, and the Self-coding scheme with one relay and the direct link, which has been shown in \cite{Yiliu} to achieve the diversity order of two, are also included in this simulation for comparison. The results are shown in Fig.~\ref{fig7}, from which we can see that the BERs of the two proposed schemes are a bit larger than the HD scheme and the 3 by 1 MIMO is the best. However, the proposed schemes have better spectrum efficiency than the HD scheme and they can achieve the same diversity as the 3 by 1 MIMO, that is, they all achieve the diversity order of three in this case.
\begin{figure}[htbp]
\centering
\includegraphics[scale=0.55]{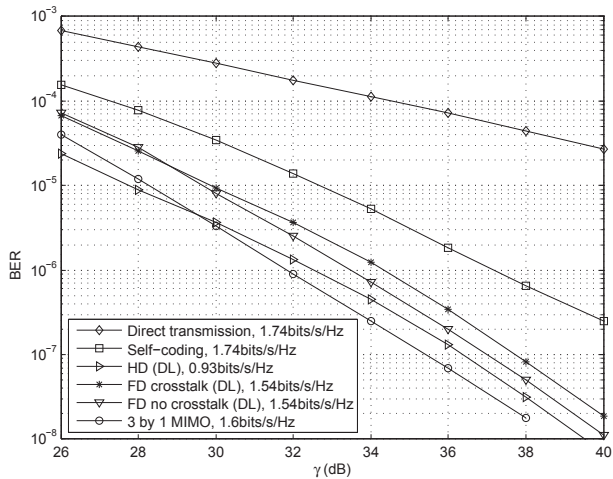}
\caption{\label{fig7} Diversity comparison with the direct link (DL) included.}
\end{figure}

\setcounter{equation}{54}
\begin{figure*}[!t]
\normalsize
\setcounter{MYtempeqncnt}{\value{equation}}
\begin{equation}\label{equ-ap2}
\begin{array}{rcl}
\mathbf{\hat{C}}&=&\left[\begin{array}{cc}&[\mathbf{\tilde{a}}_1*\mathbf{x}+K\mathbf{\tilde{a}}_1(\xi+2)*\mathbf{x}+K^2\mathbf{\tilde{a}}_1(2(\xi+2))*\mathbf{x}+\cdots]_{N+p}\\
&[\mathbf{\tilde{a}}_2*\mathbf{x}+K\mathbf{\tilde{a}}_2(\xi+2)*\mathbf{x}+K^2\mathbf{\tilde{a}}_2(2(\xi+2))*\mathbf{x}+\cdots]_{N+p}\end{array}\right]\\
&=&\left[\left\{\left[\begin{array}{cc}\mathbf{\tilde{a}}_1\\ \mathbf{\tilde{a}}_2\end{array}\right]+\left[\begin{array}{cc} K\mathbf{\tilde{a}}_1(\xi+2)\\K\mathbf{\tilde{a}}_2(\xi+2)\end{array}\right]+\left[\begin{array}{cc}K^2\mathbf{\tilde{a}}_1(2(\xi+2))\\K^2\mathbf{\tilde{a}}_2(2(\xi+2))\end{array}\right]+\cdots\right\}*\mathbf{x}
\right]_{N+p}
\end{array}
\end{equation}
\begin{equation}\label{equ-ap3}
\begin{array}{rcl}\mathbf {B}(\mathbf{\hat{C}}_1,\mathbf{\hat{C}}_2)&=&\left[\left\{\left[\begin{array}{cc}\mathbf{\tilde{a}}_1\\\mathbf{\tilde{a}}_2\end{array}\right]+\left[\begin{array}{cc} K\mathbf{\tilde{a}}_1(\xi+2)\\K\mathbf{\tilde{a}}_2(\xi+2)\end{array}\right]+\left[\begin{array}{cc}K^2\mathbf{\tilde{a}}_1(2(\xi+2))\\K^2\mathbf{\tilde{a}}_2(2(\xi+2))\end{array}\right]+\cdots\right\}*\mathbf{e}
\right]_{N+p}\end{array}
\end{equation}
\setcounter{equation}{\value{MYtempeqncnt}}
\hrulefill
\end{figure*}

Note that, as mentioned earlier, although our proposed  partial DLC-STC and DLC-STC have generator coefficient matrices SFR, we are not able to theoretically prove that they achieve the full diversity as what is done in \cite{Guo}.  However, from our simulations above, one can clearly see that they indeed achieve the full diversity with the MMSE-DFE receiver numerically.

\section{Conclusions}\label{sec7}
If AF protocol is adopted in two-relay asynchronous full-duplex cooperative communication networks with cross-talks, we showed that the cross-talk interference cannot be removed well. We then first proposed a partial DLC-STC scheme to make use of the cross-talks instead of removing them. For the case of two-relay asynchronous full-duplex cooperative networks without cross-talks between the relays, we also proposed a DLC-STC scheme by making use of signal from the loop channels. We showed that by controlling the amplifying factors, both schemes can achieve full asynchronous cooperative diversity when suboptimal receivers such as MMSE-DFE receivers are used. The proposed schemes can also be extended to the case where the direct link is available, by which one more diversity order can be achieved.

\begin{center}
{\bf Acknowledgement}
\end{center}
The authors would like to thank the editor and the anonymous reviewers for their careful reading of this manuscript and for their many detailed, constructive, and
useful comments and suggestions that have improved the presentation of this
paper.

\appendices
\section{Proof for Lemma~\ref{lemma1-0}}
Denoting ${\mathbf{\tilde{a}}}_k=\left[{\mathbf{\bar{a}}}_k\,\, 0\,\, 0\cdots\right]$ and its shifted version ${\mathbf{\tilde{a}}}_k(n)=\left[\mathbf{0}_{1\times n}\,\,{\mathbf{\bar{a}}}_k\,\, 0\,\cdots\right]$, where $\bar{\mathbf{a}}_k=\left[a_{k1} \,\,\,\, \mathbf{0}_{\psi_k} \,\,\,\, a_{k2} \,\,\,\, \mathbf{0}_{\xi-\psi_k}\right]$ and $k=1,2$, we have $\mathbf{\hat{a}}_k=\mathbf{\tilde{a}}_k+K\mathbf{\tilde{a}}_k(\xi+2)+K^2\mathbf{\tilde{a}}_k(2(\xi+2))+\cdots$. Because the convolution is a linear operation, we obtain
\setcounter{equation}{53}
\begin{equation}\label{equ-ap1}
[\mathbf{\hat{a}}_k*\mathbf{x}]_{N+p}\!\!=\!\![\mathbf{\tilde{a}}_k*\mathbf{x}\!+\!K\mathbf{\tilde{a}}_k(\xi+2)*\mathbf{x}\!+\!K^2\mathbf{\tilde{a}}_k(2(\xi+2))*\mathbf{x}\!+\cdots]_{N+p}
\end{equation}
The codeword of the partial DLC-STC can be written as (\ref{equ-ap2}) at the top of the next page,
where $\left[\begin{array}{cc}\mathbf{{v}}_1\\\mathbf{{v}}_2\end{array}\right]*\mathbf{x}=\left[\begin{array}{cc}\mathbf{{v}}_1*\mathbf{x}\\\mathbf{{v}}_2*\mathbf{x}\end{array}\right]$.

For any pair of distinct codewords $\mathbf{\hat{C}}_1$ and $\mathbf{\hat{C}}_2$, which are generated from two different signal blocks $\mathbf{{x}}_1$ and $\mathbf{{x}}_2$, respectively, we define the matrix $\mathbf {B}(\mathbf{\hat{C}}_1,\mathbf{\hat{C}}_2)=\mathbf{\hat{C}}_1-\mathbf{\hat{C}}_2$. It can be written as (\ref{equ-ap3}),
where $\mathbf{e}=\mathbf{x}_1-\mathbf{x}_2\neq\mathbf{0}$.

By the rank criterion of Rayleigh space-time codes \cite{tarokh98}, we know that the diversity of the code depends on the minimum rank of $\mathbf {B}(\mathbf{\hat{C}}_1,\mathbf{\hat{C}}_2)$. In (\ref{equ-ap3}), it is not hard to see that $\mathbf {B}(\mathbf{\hat{C}}_1,\mathbf{\hat{C}}_2)$ is a result of column linear transform on the matrix $\mathbf {B}(\mathbf{\tilde{C}}_1,\mathbf{\tilde{C}}_2)=\left[\left[\begin{array}{cc}\mathbf{\tilde{a}}_1\\\mathbf{\tilde{a}}_2\end{array}\right]*\mathbf{e}\right]_{N+p}$, for which $\mathbf {B}(\mathbf{\hat{C}}_1,\mathbf{\hat{C}}_2)$ and $\mathbf {B}(\mathbf{\tilde{C}}_1,\mathbf{\tilde{C}}_2)$ have the same rank.

 Next, we will show $\mathbf {B}(\mathbf{\tilde{C}}_1,\mathbf{\tilde{C}}_2)=\left[\left[\begin{array}{cc}\mathbf{\tilde{a}}_1\\\mathbf{\tilde{a}}_2\end{array}\right]*\mathbf{e}\right]_{N+p}$ and $\mathbf {B}(\mathbf{\bar{C}}_1,\mathbf{\bar{C}}_2)=\left[\begin{array}{cc}\mathbf{\bar{a}}_1\\\mathbf{\bar{a}}_2\end{array}\right]*\mathbf{e}$ have the same rank iff $p\geq \xi+1$.

Recalling ${\mathbf{\tilde{a}}}_k=\left[{\mathbf{\bar{a}}}_k\,\, 0\,\, 0\cdots\right]$, where the length of vector ${\mathbf{\bar{a}}}_k$ is $\xi+2$, we obtain
 \setcounter{equation}{56}
 \begin{equation}
 \mathbf {B}(\mathbf{\bar{C}}_1,\mathbf{\bar{C}}_2)=\left[\begin{array}{cc}\mathbf{\bar{a}}_1\\\mathbf{\bar{a}}_2\end{array}\right]*\mathbf{e}=\left[\begin{array}{cc}\mathbf{\bar{a}}_1*\mathbf{e}\\\mathbf{\bar{a}}_2*\mathbf{e}\end{array}\right]
 \end{equation}

 \begin{equation}
 \mathbf {B}(\mathbf{\tilde{C}}_1,\mathbf{\tilde{C}}_2)=\left[\left[\begin{array}{cc}\mathbf{\tilde{a}}_1\\\mathbf{\tilde{a}}_2\end{array}\right]*\mathbf{e}\right]_{N+p}\!\!\!\!\!=\!\left[\begin{array}{cc}\left[{\mathbf{\bar{a}}}_1\,\, 0\,\, 0\cdots\right]*\mathbf{e}\\\left[{\mathbf{\bar{a}}}_2\,\, 0 \,\, 0 \cdots\right]*\mathbf{e}\end{array}\right]_{N+p}
 \end{equation}
Since the number of columns in $\mathbf {B}(\mathbf{\bar{C}}_1,\mathbf{\bar{C}}_2)$ is $N+\xi+1$, it is obvious that $R(\mathbf{B}(\mathbf{\tilde{C}}_1,\mathbf{\tilde{C}}_2))=R(\mathbf {B}(\mathbf{\bar{C}}_1,\mathbf{\bar{C}}_2))$ if the number of columns in $\mathbf {B}(\mathbf{\tilde{C}}_1,\mathbf{\tilde{C}}_2)$ is no less than that in $\mathbf {B}(\mathbf{\bar{C}}_1,\mathbf{\bar{C}}_2)$, i.e., $p\geq \xi+1$, where $R(\mathbf{B})$ stands for the rank of matrix $\mathbf{B}$.

 While, if $p< \xi+1$, which means $\mathbf {B}(\mathbf{\tilde{C}}_1,\mathbf{\tilde{C}}_2)$ is the submatrix formed by the first $N+p$ columns in $\mathbf {B}(\mathbf{\bar{C}}_1,\mathbf{\bar{C}}_2)$, it cannot be ensured that the two matrices have the same rank since $\mathbf{\bar{a}}_1$, $\mathbf{\bar{a}}_2$, and $\mathbf{e}$ are all random.

From the discussion above, we know that iff $p\geq \xi+1$, $R(\mathbf{B}(\mathbf{\tilde{C}}_1,\mathbf{\tilde{C}}_2))=R(\mathbf {B}(\mathbf{\bar{C}}_1,\mathbf{\bar{C}}_2))$ for any two different source frames $\mathbf{x}_1$ and $\mathbf{x}_2$. Considering $R(\mathbf{B}(\mathbf{\hat{C}}_1,\mathbf{\hat{C}}_2))=R(\mathbf {B}(\mathbf{\tilde{C}}_1,\mathbf{\tilde{C}}_2))$, we achieve $R(\mathbf{B}(\mathbf{\hat{C}}_1,\mathbf{\hat{C}}_2))=R(\mathbf {B}(\mathbf{\bar{C}}_1,\mathbf{\bar{C}}_2))$ for any two different source frames $\mathbf{x}_1$ and $\mathbf{x}_2$ iff $p\geq \xi+1$.

So the partial DLC-STC $\mathbf{\hat{C}}$ achieves the same asynchronous diversity as the DLC-STC $\mathbf{\bar{C}}$ iff $p\geq \xi+1$.


\begin{thebibliography}{10}
\providecommand{\url}[1]{#1}
\csname url@samestyle\endcsname
\providecommand{\newblock}{\relax}
\providecommand{\bibinfo}[2]{#2}
\providecommand{\BIBentrySTDinterwordspacing}{\spaceskip=0pt\relax}
\providecommand{\BIBentryALTinterwordstretchfactor}{4}
\providecommand{\BIBentryALTinterwordspacing}{\spaceskip=\fontdimen2\font plus
\BIBentryALTinterwordstretchfactor\fontdimen3\font minus
  \fontdimen4\font\relax}
\providecommand{\BIBforeignlanguage}[2]{{%
\expandafter\ifx\csname l@#1\endcsname\relax
\typeout{** WARNING: IEEEtran.bst: No hyphenation pattern has been}%
\typeout{** loaded for the language `#1'. Using the pattern for}%
\typeout{** the default language instead.}%
\else
\language=\csname l@#1\endcsname
\fi
#2}}
\providecommand{\BIBdecl}{\relax}
\BIBdecl

\bibitem{Politis}
C. Politis, T. Oda, S. Dixit, A. Schieder, H.-Y. Lach, M.I. Smirnov, S. Uskela, and R. Tafazolli, ``Cooperative networks for the future wireless world,'' \emph{IEEE Commun. Mag.}, vol. 42, no. 9, pp. 70--79, Sept. 2004.

\bibitem{Sendonaris01}
{A. Sendonaris, E. Erkip, and B. Aazhang}, ``User cooperation diversity-{P}art {I}. System description,'' \emph{IEEE Trans. Commun.}, vol.~51, no.~11, pp. 1927--1938, Nov. 2003.

\bibitem{Sendonaris02}
------, ``User cooperation diversity-{P}art {II}: implementation aspects and
  performance analysis,'' \emph{IEEE Trans. Commun.}, vol.~51, no.~11, pp.
  1939--1948, Nov. 2003.

\bibitem{Laneman02}
{J. N. Laneman, D. N. C. Tse, and G. W.Wornell}, ``Cooperative diversity in
  wireless networks: efficient protocols and outage behavior,'' \emph{IEEE
  Trans. Inf. Theory}, vol.~50, no.~12, pp. 3062--3080,  Dec. 2004.
\bibitem{6020878}
Y.-W. Liang and R. Schober, ``Cooperative amplify-and-forward beamforming with multiple multi-antenna relays,'' \emph{IEEE Trans. Commun.}, vol. 59, no. 9, pp.2605--2615, Sept. 2011.

\bibitem{Choi}
C.-H. Choi, U.-K. Kwon, Y.-J. Kim, and G.-H. Im, ``Spectral efficient
  cooperative diversity technique with multi-layered modulation,''
  \emph{IEEE Trans. Commun.}, vol.~58, no.~12, pp. 3480--3490, Dec. 2010.

\bibitem{1495421}
T.~Himsoon, W.~Su, and K.~Liu, ``Differential transmission for
  amplify-and-forward cooperative communications,'' \emph{IEEE Signal Process. Lett.}, vol.~12, no.~9, pp. 597--600, Sept. 2005.

\bibitem{5606181}
C.~Wang, Y.~Fan, J. S.~Thompson, M.~Skoglund, and H. V.~Poor, ``Approaching the
  optimal diversity-multiplexing tradeoff in a four-node cooperative network,''
  \emph{IEEE Trans. Wireless Commun.}, vol.~9, no.~12, pp.
  3690--3700, Dec. 2010.

\bibitem{5450057}
J.-S. Yoon, J.-H. Kim, S.-Y. Yeo, and H.-K. Song, ``Performance analysis of
  mutual decode-and-forward cooperative communication for OFDMA based mobile
  WiMAX system,'' in \emph{ Proc. 2009 IEEE International Symposium on Personal, Indoor and Mobile Radio Communications}, pp. 2748--2751.

\bibitem{Laneman01}
{J. N. Laneman and G. W.Wornell}, ``Distributed space-time-coded protocols for
  exploiting cooperative diversity in wireless networks,'' \emph{IEEE Trans.
  Inf. Theory}, vol.~49, no.~10, pp. 2415--2425, Oct. 2003.

\bibitem{Yindi}
Y. Jing and B. Hassibi, ``Distributed space-time coding in
wireless relay networks,'' \emph{IEEE Trans. Wireless Commun.}, vol.~5, no~12, pp. 3524--3536, Dec. 2006.

\bibitem{4063541}
Y.~Ding, J.-K. Zhang, and K.~M. Wong, ``The amplify-and-forward half-duplex
  cooperative system: pairwise error probability and precoder design,''
  \emph{IEEE Trans. Signal Process.}, vol.~55, no.~2, pp. 605--617,
  Feb. 2007.

\bibitem{1321222}
R.~Nabar, H.~Bolcskei, and F.~Kneubuhler, ``Fading relay channels: performance
  limits and space-time signal design,'' \emph{IEEE J. Sel. Areas Commun.}, vol.~22, no.~6, pp. 1099--1109, Aug. 2004.

\bibitem{5161790}
T.~Riihonen, S.~Werner, R.~Wichman, and Z.~Eduardo, ``On the feasibility of
  full-duplex relaying in the presence of loop interference,'' in \emph{Proc. 2009 IEEE Workshop on Signal Processing Advances in Wireless Communications, SPAWC '09 }, pp. 275--279.

\bibitem{5089955}
T.~Riihonen, S.~Werner, and R.~Wichman, ``Optimized gain control for
  single-frequency relaying with loop interference,'' \emph{IEEE Trans. Wireless Commun.}, vol.~8, no.~6, pp. 2801--2806, June 2009.

\bibitem{5961159}
------, ``Hybrid full-duplex/half-duplex relaying with transmit power
  adaptation,'' \emph{IEEE Trans. Wireless Commun.}, vol.~10,
  no.~9, pp. 3074--3085, Sept. 2011.

\bibitem{4917875}
T.~Riihonen, S.~Werner, R.~Wichman, and J.~Hamalainen, ``Outage probabilities in infrastructure-based single-frequency relay links,'' in \emph{Proc. 2009 IEEE Wireless Communications and Networking Conference}, pp.1--6.

\bibitem{5470111}
T.~Riihonen, S.~Werner, and R.~Wichman, ``Spatial loop interference suppression
  in full-duplex MIMO relays,'' in \emph{Proc. 2009 Asilomar Conference on Signals, Systems and Computers}, pp. 1508--1512.

\bibitem{6069620}
T. Riihonen, S. Werner, R. Wichman, ``Mitigation of loopback self-interference in full-duplex MIMO relays,'' \emph{IEEE Trans. Signal Process.}, vol.~59, no.~12, pp.5983--5993, Dec. 2011.

\bibitem{1499041}
G.~Kramer, M.~Gastpar, and P.~Gupta, ``Cooperative strategies and capacity
  theorems for relay networks,'' \emph{IEEE Trans. Inf. Theory}, vol.~51, no.~9, pp. 3037--3063, Sept. 2005.

\bibitem{6151271}
A. Sahai, G. Patel, and A. Sabharwal, ``Asynchronous full-duplex wireless,'' in \emph{Proc. 2012 Fourth International Conference on Communication Systems and Networks}, pp.1--9.

\bibitem{Yiliu}
Y. Liu, X.-G. Xia, and H.-L. Zhang, ``Distributed space-time coding for full-duplex asynchronous cooperative communications,'' \emph{IEEE Trans. Wireless Commun.}, vol.~11, no.~7, pp. 2680--2688, July 2012.

\bibitem{Guo}
X.~Guo and X.-G. Xia, ``Distributed linear convolutive space-time codes for
  asynchronous cooperative communication networks,'' \emph{IEEE Trans. Wireless Commun.}, vol.~7, no.~5, pp. 1857--1861, May 2008.

\bibitem{4155128}
Y.~Li and X.-G. Xia, ``A family of distributed space-time trellis codes with
  asynchronous cooperative diversity,'' \emph{IEEE Trans. Commun.}, vol.~55, no.~4, pp. 790--800, Apr. 2007.

\bibitem{Shang001}
Y.~Shang and X.-G. Xia, ``Shift-full-rank matrices and applications in
  space-time trellis codes for relay networks with asynchronous cooperative
  diversity,'' \emph{IEEE Trans. Inf. Theory}, vol.~52, no.~7,
  pp. 3153--3167, July 2006.


\bibitem{5683902}
Y.~Rong, ``Non-regenerative multi-hop MIMO relays using MMSE-DFE technique,''
  in \emph{Proc. 2010 GLOBECOM}, pp. 1--5.

\bibitem{869048}
N.~Al-Dhahir and A.~Sayed, ``The finite-length multi-input multi-output
  MMSE-DFE,'' \emph{IEEE Trans. Signal Process.}, vol.~48, no.~10,
  pp. 2921--2936, Oct. 2000.


\bibitem{tarokh98}
V. Tarokh, N. Seshadri, and A. R. Calderbank, ``Space-time
codes for high data rate wireless communication: performance
criterion and code construction,"
\emph{IEEE Trans. Inf. Theory},
 vol. 44, no. 2, pp. 744--765, Mar. 1998.


\bibitem{1523684}
J.-K. Zhang, J.~Liu, and K.~M. Wong, ``Linear Toeplitz space time block
  codes,'' in \emph{Proc. 2005 ISIT }, pp. 1942--1946.
\end{thebibliography}
\end{document}